\documentclass[conference]{IEEEtran}
\IEEEoverridecommandlockouts
\usepackage{cite}
\usepackage{amsmath,amssymb,amsfonts}
\usepackage{amsthm}
\usepackage{algorithmic}
\usepackage{graphicx}
\usepackage{textcomp}
\usepackage{xcolor}
\def\BibTeX{{\rm B\kern-.05em{\sc i\kern-.025em b}\kern-.08em
    T\kern-.1667em\lower.7ex\hbox{E}\kern-.125emX}}
\usepackage{booktabs} 

\usepackage[lined, boxed, linesnumbered, commentsnumbered, ruled]{algorithm2e}
\usepackage{mathtools}
\usepackage{xcolor}
\usepackage{amssymb} 
\usepackage{bbm}

\usepackage{enumitem}
\graphicspath{{figure/}{figures/}}
\newtheorem{theorem}{Theorem}[section]

\newtheorem{corollary}{Corollary}[theorem]
\newtheorem{lemma}[theorem]{Lemma}
\newtheorem{definition}[theorem]{Definition}
\newtheorem{assumption}[theorem]{Assumption}

\usepackage{lipsum}
\usepackage{graphicx}
\ifCLASSOPTIONcompsoc
    \usepackage[caption=false, font=normalsize, labelfont=sf, textfont=sf]{subfig}
\else
\usepackage[caption=false, font=footnotesize]{subfig}
\fi
\setlist[itemize]{leftmargin=*}
\allowdisplaybreaks

\begin{document}

\title{Is Deadline Oblivious Scheduling Efficient\\ for Controlling Real-Time Traffic\\ in Cellular Downlink Systems? 
}

\author{
    \IEEEauthorblockN{Sherif ElAzzouni\IEEEauthorrefmark{1}, Eylem Ekici\IEEEauthorrefmark{1}, Ness Shroff\IEEEauthorrefmark{1}\IEEEauthorrefmark{2}}
    \IEEEauthorblockA{\IEEEauthorrefmark{1}Dept. of Electrical and Computer Engineering, Ohio State University.}
    \IEEEauthorblockA{\IEEEauthorrefmark{2}Dept. of Computer Science and Engineering, Ohio State University.
    \\\{elazzouni.1, ekici.2, shroff.11\}@osu.edu}
}

\maketitle

\begin{abstract}
The emergence of bandwidth-intensive latency-critical traffic in 5G Networks, such as Virtual Reality and Cloud Gaming, has motivated interest in wireless resource allocation problems for flows with hard-deadlines. Attempting to solve this problem brings about the following two key challenges: (i) The flow arrival and the wireless channel state information are not known to the Base Station (BS) apriori, thus, the allocation decisions need to be made in an \textit{online} manner. (ii) Resource allocation algorithms that attempt to maximize a reward in the wireless setting will likely be unfair, causing unacceptable service for some users. In the first part of this paper, we model the problem of allocating resources to deadline-sensitive traffic as an online convex optimization problem, where the BS acquires a per-request reward that depends on the amount of traffic transmitted within the required deadline. We address the question of whether we can efficiently solve that problem with low complexity. In particular, whether we can design a constant-competitive scheduling algorithm that is oblivious to requests' deadlines. To this end, we propose a primal-dual Deadline-Oblivious (DO) algorithm, and show it is approximately 3.6-competitive. Furthermore, we show via simulations that our algorithm tracks the prescient offline solution very closely, significantly outperforming several algorithms that were previously proposed.  Our results demonstrate that even though a scheduler may not know the deadlines of each flow, it can still achieve good theoretical and empirical performance. In the second part, we impose a stochastic constraint on the allocation, requiring a guarantee that each user achieves a certain timely throughput (amount of traffic delivered within the deadline over a period of time). We propose a modified version of our algorithm, called the Long-term Fair Deadline Oblivious (LFDO) algorithm for that setup. We combine the Lyapunov framework for stochastic optimization with the Primal-Dual analysis of online algorithms, to show that LFDO retains the high-performance of DO, while satisfying the long-term stochastic constraints.

\end{abstract}


\section{Introduction}
Next generation mobile networks are poised to support a set of diverse applications, many of which are both bandwidth-intensive and latency-sensitive, having strict requirements on end-to-end delay. In applications like Virtual Reality, Cloud Gaming, and Video Streaming, it is critical that end users receive the bulk of their data within a prespecified hard deadline. Any extra delay would usually render the transmission useless. On the other hand, the high bandwidth requirements of those applications would often make streaming all users' data within the deadline impossible, thus, a good scheduler has to balance those two goals, intelligently making decisions on how to use the available bandwidth to maximize end users' satisfaction. This motivates the design of resource allocation schemes that jointly account for bandwidth requirements, hard deadlines and applications' priorities in terms of what has to be transmitted to end-users to maintain a seamless experience.\\ 
To model the problem of resource allocation and scheduling for bandwidth-intensive latency-critical applications, we propose approaching the problem as an online scheduling problem, where requests arrive to the BS carrying a payload, a hard deadline, and a concave reward function that rewards successful partial transmission within the prespecified hard deadline. Our motivations is that, in many applications, completing a request partially within a deadline is acceptable. For example in video transmission, frame-dropping and error concealment are used to adapt to lower bandwidths, thus, this fits our model where 1. transmitting a frame after the deadline is useless, 2. the portion of the request completed exhibits a diminishing return. Another example is VR applications and/or $360^{\circ}$ videos where tiles outside field-of-view can be adaptively streamed at a lower rate if needed \cite{hosseini2016adaptive}. A third example is mobile cloud gaming, where the cloud server adaptively transmit most-likely sequences depending on the bandwidth availability \cite{lee2015outatime}, thus, also an example of a high-bandwidth hard deadline application with diminishing returns.\\
Having modeled our problem as an online scheduling problem, the central question becomes ``\textbf{Can we find a constant-competitive solution that has low-complexity?}". Specifically, we are interested in the class of ``\textbf{deadline-oblivious}" algorithms, that make scheduling decisions without taking individual flows' deadline requirements into account. Those algorithms have low complexity, are more amenable to implementation than deadline-aware schedulers, and are robust against deadline information absence or inaccuracy.\\   
We show that the answer to this question is affirmative. Our solution to the problem follows the online primal-dual approach presented in \cite{buchbinder2009design} for online linear programs and used in \cite{zheng2016online} \cite{lucier2013efficient} \cite{devanur2018primal} in the context of datacenter scheduling. The problem of online deadline-sensitive scheduling in wireless networks presents the following unique challenges: 1. Time-varying complex non-orthogonal capacity regions due to the nature of the wireless channel, and a set of power control, coding and MIMO capabilities, that a Base Station (BS) can use to achieve rates within the capacity region. Our problem formulation treats instantaneous capacity region as a time-varying closed convex region with no assumptions on the orthogonality of user rates.  2. Susceptibility of opportunistic scheduling to unfairness, as any utility-maximizing algorithm would prefer users with consistently good channels. We tackle long-term unfairness through stochastic timely-throughput constraints. Our key contributions can be summarized as follows: 
\begin{enumerate}[wide, labelwidth=!, labelindent=0pt]
\setlength{\itemsep}{0.mm}
\item We develop a Primal-Dual Deadline Oblivious (DO) algorithm to solve the problem of scheduling deadline sensitive traffic, and show in Theorem \ref{key}, that our online solution provides a 3.6 competitive ratio compared to the offline prescient solution that has all the information apriori.
\item We show in Theorem \ref{LFDOperf} that the  Primal-Dual algorithm can be modified to satisfy long-term stochastic ``Timely Throughput" constraints. Timely throughput is the amount of traffic delivered to the end user within the allowed deadline over a certain time period. We show that this modification causes minimal sacrifice to performance by utilizing a virtual queue structure and Lyapunov arguments in a novel way.
\item We show via simulations that our algorithm outperforms some well-known algorithms proposed in the literature for deadline-sensitive traffic scheduling. We also show that our algorithm closely tracks the offline optimal solution. Furthermore, we verify the efficacy of the modified Long-term Fair Deadline Oblivious (LFDO) algorithm in satisfying timely throughput constraints.
\end{enumerate} 
Online Scheduling of Deadline-constrained traffic is a classical problem in networking \cite{pruhs2004online}. This problem has received increased recent attention with the proliferation of deadline-sensitive applications in datacenters. A preemptive algorithm that relies on the slackness metric was proposed in \cite{lucier2013efficient}. In \cite{devanur2018primal}, it was shown that online primal-dual algorithms are also energy efficient. Perhaps closest to our setup is the work in \cite{zheng2016online}, where hard-deadlines and partial utilities are considered for multi-resource allocation. We compare our algorithm to the one in \cite{zheng2016online} in the simulation section and show that our algorithm has better performance due to reliance on primal-dual updates rather than only primal updates. The aforementioned works however do not take into account the fundamental challenges of the wireless setup that we have discussed.\\    
In the wireless setting, there has been an increasing interest in deadline-constrained traffic. In particular, the concept of ``timely-throughput" has been proposed and studied extensively \cite{hou2014scheduling} \cite{lashgari2013timely} \cite{shakkottai2002scheduling} for packets with deadlines. However, these works target packet transmissions and do not consider the ``diminishing returns" properties of bandwidth-intensive traffic at the flow level. 

\section{System Model}
\begin{figure}
\centering
\includegraphics[height=2.8cm,width=6cm]{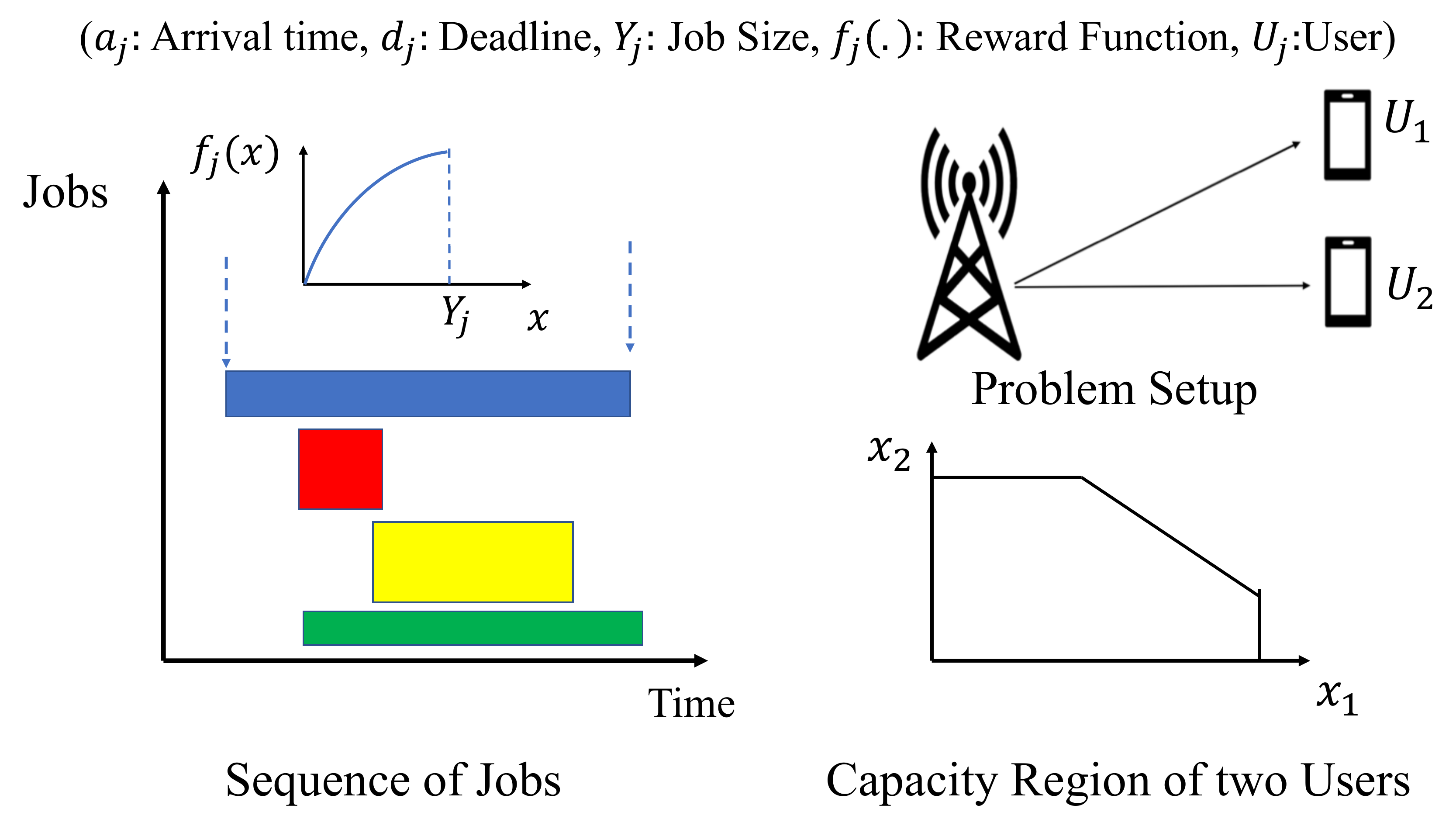}
\caption{System Model}
\vspace{-0.5cm}
\end{figure}

The system model is shown in Fig.1. Every time slot, we model every job/request $j$ arriving at the BS as the tuple $(a_j,d_j,Y_j,f_j(.), U_j)$, representing arrival time, deadline, job size, concave reward function that rewards the amount of the job served $x$ with $f_j(x)$, and an intended user among an available $N$ users, that is, $U_j \in \{1,2, \ldots, N\}$.\\
 At each time slot, $t$, the BS calculates an instantaneous feasible rate region $\mathbf{R}[t]$, based on the CSI feedback. The feasible rate region determines the rates that the BS can allocate to different users at each time slot. We do not make any assumptions on $\mathbf{R}[t]$, except that it is closed, bounded, and convex. We model the feasible rate regions over time in this way to capture both the time variability characteristic of wireless networks as well as the BS capabilities to employ power control, coding, and MIMO to extend the rate region beyond the simple orthogonal capacity region (see for example \cite{dai2015non}). We remark that this assumption changes the problem significantly from the typical datacenter job-resource pairing (e.g. \cite{zheng2016online}), where the capacity is assumed to be orthogonal with no time-variation. \\
 Each job $j$ is active between its arrival time, $a_j$, and its deadline $d_j$, after which the job expire and no reward would be gained from transmitting it. At each time slot $t$, each active job $j$ is allocated a rate $x_{tj}$. We use the variable $A_{tj}$ as an indicator of whether a job $j$ is active at time $t$. We collect those indicators at time $t$ in a diagonal matrix that we refer to as $\mathbf{A}_t$. We denote all the jobs that arrive over the problem horizon by the set $J$, and all rates given to all jobs at time $t$ by $\mathbf{x}_t=(x_{t1}, x_{t2}, \ldots, x_{tJ})$.\\
 We assume that utility functions $f_j(.)$ are continuous, strictly concave, non-decreasing, and differentiable with a gradient $\partial f_j(.)$ and $f_j(0)=0$ for all jobs $j$. This captures the diminishing return properties of the job service. With some abuse of notation we will refer to the vector of the gradients of all functions as $\nabla f(\,)=( \partial f_1(\,), \partial f_2(\,), \ldots, \partial f_J(\,))$.
 
 \section{Problem Formulation} 
 We model the problem as a finite-horizon online convex optimization problem aiming to maximize the total utility obtained from the total resources received by each job prior to expiry. Formally:
 
 \begin{small}
 \vspace{-0.2cm}
\begin{subequations}\label{formulation}
        \begin{align}
    & \underset{\mathbf{x}_1, \mathbf{x}_2, \ldots,\mathbf{x}_t}{\max}
    & &  \sum_{j \in  J} f_j(\sum_{t=1}^T A_{tj} x_{tj}) \label{objective} \\
    & \text{subject to}
    & & \sum_{t=1}^T x_{tj} \leq Y_j, \;  \forall j \label{budget} \\ 
    &&& \mathbf{x}_t \in \mathbf{R}[t], \;  \;  \forall t= 1,2, \ldots, T. \label{capacity}
\end{align}
    \end{subequations}
    \end{small}   
  The objective function \eqref{objective} is the utility achieved by each job, due to the sum of resources allocated to that job over its activity window. The constraint \eqref{budget} ensures jobs are not allocated more than their size. The constraint \eqref{capacity} ensures that the rates allocated by the BS are feasible w.r.t the rate region estimated from the CSI feedback. Technically, this constraint should be on the users rates, not on the jobs. However, it is easy to transform the constraints on users' sum rates to constraints on individual jobs, since every job has a single intended user.\\ 
  Our performance metric throughout will be the \textit{Competitive Ratio (CR)}. The Competitive Ratio, $\gamma$, guarantees that the online algorithm always achieves at least, a $\frac{1}{\gamma}$ fraction of the total reward achieved by an optimal offline prescient solution that knows all jobs' details before-hand as well as all the rate regions, independent of the problem size. Denote the total reward achieved by an online algorithm as $P=\sum_{j \in  J} f_j(\sum_{t=1}^T A_{tj} x_{tj})$. We call the offline optimal algorithm OPT, and denote the total reward achieved by OPT as $P^{*}=\sum_{j \in  J} f_j(\sum_{t=1}^T A_{tj} x_{tj}^{*})$
  \begin{definition}{Competitive Ratio: } An online algorithm is $\gamma$-competitive if the following holds:
  \begin{small}
\begin{equation}
\gamma \leq \sup_{\mathbf{S}_j, R[1], R[2], \ldots, R[T]} \frac{P^{*}}{P}
\end{equation}
  \end{small}
  where $\mathbf{S}_j$ is the input job sequence over all slots.
  \end{definition}
  \subsection*{Dual Problem}
  Since our solution is based on simultaneously updating the primal and dual solutions, we start by deriving the dual optimization problem:
\begin{small}
\begin{subequations}\label{formulation}
        \begin{align}
    & \underset{\alpha, \beta}{\min}
    & &  \sum_{t=1}^T \underset{\mathbf{x}_t \in \mathbf{R}[t]}{\text{max}} <A_t \mathbf{\alpha} -\mathbf{\beta},\mathbf{x}> + \mathbf{\beta}^TY -\sum_{j=1}^J f_j^*(\alpha_j)\label{dual} \\
    & \text{subject to}
    & & \mathbf{\alpha}, \mathbf{\beta} \geq 0  
\end{align}
\end{subequations}
\end{small}
where $\mathbf{\alpha}=[\alpha_1, \ldots, \alpha_J]$ is the $J\times 1$ Fenchel Dual vector , $\mathbf{\beta}=[\beta_1, \ldots, \beta_J]$ is the $J\times1$ multiplier of the constraint \eqref{budget}, and $Y=[Y_1, \ldots, Y_J]$. The operator $<\, , \,>$ is the inner product operator. The function $f_j^*(\alpha_j)$ is the cocave conjugate of the function $f_j(\,)$ \cite{boyd2004convex}, which can be written as:

\begin{small}
\begin{equation}\label{conj1}
f^*_j(\alpha_j)= \inf_{x \geq 0} <\alpha_j,x> - f_j(x)
\end{equation} 
\end{small}
A solution $(\mathbf{x},\mathbf{\alpha}, \mathbf{\beta})$ is a primal-dual solution if and only if:
\begin{small}
\begin{equation*} 
\mathbf{x}_t =  \underset{\mathbf{x} \in R[t]}{\text{argmax}} <A_t \mathbf{\alpha} -\mathbf{\beta},\mathbf{x}>, \, \, \,\,  \, \,\,
\alpha_j =  \partial(f_j(\sum_{t=1}^T A_{tj} x_{tj})).
\end{equation*}
\end{small}
To derive a Competitive Ratio bound for our algorithm, we use the following theorem on primal and dual problems:
\begin{theorem}{(Weak and Strong Duality \cite{boyd2004convex})} \label{dualityThm}
Let $(\mathbf{x}_1, \ldots, \mathbf{x}_{T})$ and $(\mathbf{\alpha}, \mathbf{\beta})$ be feasible solutions for the Primal and the Dual problems respectively, then the following holds:
\begin{small}
\begin{equation}
D=\sum_{t=1}^T \sigma_t(\mathbf{\alpha},\mathbf{\beta}) + \mathbf{\beta}^TY -\sum_{j=1}^J f_j^*(\alpha_j) \geq \sum_{j \in  J} f_j(\sum_{t=1}^T A_{tj} x_{tj}) =P
\end{equation}
\end{small}
where $\sigma_t(\mathbf{\alpha},\mathbf{\beta}) =\underset{\mathbf{x} \in \mathbf{R}[t]}{\text{max}} <A_t \mathbf{\alpha} -\mathbf{\beta},\mathbf{x}>$.\\
For the optimal offline Primal and Dual solutions, assuming strong duality, the following holds:
\begin{small}
\begin{equation}
D\geq  D^*=P^* \geq P \label{duality}
\end{equation}
\end{small}
\end{theorem}
This gives us a method to bound the competitive ratio of any primal-dual online algorithm by showing that $D \leq \gamma P$, which implies that $P^* \leq D \leq \gamma P$. This technique is covered in depth for online linear programs in \cite{buchbinder2009design} (e.g. Theorem 2.3), for many applications. We use the same idea to analyze our online algorithm presented in the next section.
\section{Deadline Oblivious (DO) Algorithm}
\subsection{Algorithm}
Before presenting our algorithm, we give some intuition on how we developed it. It is useful to think of our problem as an online fractional matching problem with edge weights on a bipartite graph. One side of the graph are the jobs, and on the other side are the time slots. Each time slot brings new information on the capacity, edge weights, and utility functions. It is well known that for the simplest online matching problem with linear rewards, there exists an $e-1$-competitive Primal-Dual algorithm that outperforms the simple Greedy Algorithm that is $2$-competitive \cite{mehta2013online}. Later, this framework was extended for concave reward functions for covering/packing problems \cite{azar2016online}, and for online matching problems \cite{eghbali2016designing}. In fact, our algorithm builds on the  algorithm presented in \cite{eghbali2016designing} for online matching with capacity constraints only and no job size constraints. We develop a complete resource allocation algorithm for deadline sensitive traffic with job sizes constraints, as well as tackling the long-term stochastic constraints.
\begin{small}
\begin{algorithm}
\newcommand{\forcond}{$t=1$ \KwTo $T$}
\SetKw{Kw}{Initialize:}
\Kw{ At $t=0$, set $\beta_{tj}=0, \; \forall j$}\\ 
\For{\forcond}{
BS receives new jobs arriving at time $t$, and calculates $\mathbf{R}[t]$ \\
Calculate the pair $(\mathbf{\alpha}_t, \mathbf{x}_t)$ that solves the following saddle point problem:
\begin{small}
\begin{equation*}
\underset{\mathbf{\alpha} \geq 0}{\text{min}} \,
 \underset{\mathbf{x} \in \mathbf{R}[t]}{\text{max}} -f^{*}(\alpha) +<\alpha -\beta_{t-1}, \sum_{s=1}^{t-1} \mathbf{A}_s \mathbf{x}_s+\mathbf{A}_t\mathbf{x}>
\end{equation*}
\end{small}
Update the dual variable for every job $\mathbf{\beta}_{tj}$ as follows:
\begin{small}
\begin{multline*}
\beta_{tj}= \frac{\partial f (\sum_{s=1}^{t} A_{sj} x_{sj})}{\partial f (\sum_{s=1}^{t-1} A_{sj} x_{sj})}\bigg(1 + \frac{A_{tj}x_{tj}}{Y_j}\bigg) \beta_{t-1j}\\
 + \frac{\partial f (\sum_{s=1}^{t} A_{sj} x_{sj}) A_{tj} x_{tj}}{(C-1) Y_j}
\end{multline*}
\end{small}
}
 \caption{Deadline Oblivious (DO) Algorithm}
\end{algorithm}
\end{small}

The algorithm continuously allocates resources to active jobs by controlling $\mathbf{x}_{t}$, and updates the per-job dual variables $\mathbf{\alpha}_t=[\alpha_{t1}, \ldots, \alpha_{tJ}]$, and $\mathbf{\beta}_t=[\beta_{t1}, \ldots, \beta_{tJ}]$ every time slot accordingly. Line 4 of the algorithm jointly allocates the primal and dual variables by solving a low complexity saddle point problem. We will later show how to use approximation to further reduce the complexity of the problem. Line 5 updates the dual variable $\mathbf{\beta}$ that ensures that no job is allocated more resources than its size. This discounts the reward obtained from any job as it gets closer to completion, hence, this discounting gives priority to jobs that have more work remaining. Note that the instantaneous primal and dual allocations of all jobs do not use the knowledge of the activity window after time $t$. Since the algorithm is deadline oblivious, decisions only depend on current activity of a job and do not take into account the future activity until the deadline.\\
We define the capacity-to-file-size ratio, $F_{\max}$, as the maximum ratio between the resources any job can receive at any one time slot and the total job size. We assume, that $F_{\max}>1$, i.e., no job can be fully transmitted over one time-slot. This assumption is essential to obtain a constant competitive ratio. This is equivalent to the ``bid-to-budget" ratio assumption in online matching problems \cite{mehta2013online}. Also, let $C$ in line 5 of the algorithm be $C=(1+F_{\max})^{\frac{1}{F_{\max}}}$. Note that as $F_{\max}$ approaches zero, $C$ approaches $e$, which will be useful when we derive the competitive ratio. 
\subsection{Analysis} \label{do_analysis}
In the next few Lemmas, we will show that the DO algorithm has some useful properties that enable us to derive a relationship between the primal and dual objectives. We first define a complementary pair
\begin{definition}
$\mathbf{x}$ and $\mathbf{\alpha}$ are said to be a Complementary Pair if any one of those properties hold (It can be shown that they are all equivalent)
\begin{equation*}
f'(x)=\alpha, \, \, \, f^{*'}(\alpha)=x,  \, \, \,f(x)+f^*(\alpha)=x \alpha,
\end{equation*}

where $f^*(\alpha)$ is the concave conjugate  defined in \eqref{conj1}.
\end{definition}
\begin{lemma}\label{properties}
DO produces a primal-dual solution $(\mathbf{x,\alpha, \beta})$ that guarantees the following for all time slots:
\begin{enumerate}
\item $(\alpha_{tj}, \sum_{s=1}^t A_{sj}x_{sj})$ are a complementary pair for all time slots $t$, and for all jobs $j \in J$, i.e., $\alpha_{tj} \in \partial f_j(\sum_{s=1}^t A_{sj}x_{sj}) $
\item $\mathbf{x_t} \in \underset{\mathbf{x} \in \mathbf{R}[t]}{\text{argmax}} <\alpha -\beta, \sum_{s=1}^{t-1} \mathbf{A}_s \mathbf{x}_s+\mathbf{A}_t\mathbf{x}>$
\end{enumerate}
\end{lemma}
The Proof of the Lemma is immediate from the properties of the concave-conjugate property and the inner maximization problem in line 4 of the algorithm. The next two Lemmas ensures that DO produces a feasible primal-dual solution\\
\begin{lemma} \label{geom}
For any job $j$, the dual variable $\beta_{tj}$ grows as a geometric series that can be bounded from below as follows
\begin{small}
\begin{equation}
\beta_{tj} \geq \frac{\partial f( \sum_{s=0}^t A_{sj} x_{sj})}{C-1} (C^{\frac{\sum_{s=0}^t A_{sj} x_{sj}}{Y_j}}-1)
\end{equation}
\end{small}
\end{lemma}
\begin{proof}
We prove the Lemma by induction. The base case is $t=0$, where $\beta_{tj} \geq 0$ is trivially satisfied. Suppose the claim is true for $t-1$, then substituting in the update equation in Algorithm 1, line 5, we obtain the following:

\begin{small}
\begin{align}
\begin{split}
    \beta_{tj}= {}&  \frac{\partial f (\sum_{s=1}^{t} A_{sj} x_{sj})}{\partial f (\sum_{s=1}^{t-1} A_{sj} x_{sj})}\bigg(1 + \frac{A_{tj}x_{tj}}{Y_j}\bigg) \beta_{t-1j}\\
    & + \frac{\partial f (\sum_{s=1}^{t} A_{sj} x_{sj}) A_{tj} x_{tj}}{(C-1) Y_j}
\end{split}\\
     \overset{\mathrm{(a)}}{\geq} {}& \frac{\partial f (\sum_{s=1}^{t} A_{sj} x_{sj})}{C-1}\bigg(C^{\frac{\sum_{s=0}^{t-1} A_{sj} x_{sj}}{Y_j}}\big(1+\frac{A_{tj}x_{tj}}{Y_j}\big)-1\bigg)\\
    \overset{\mathrm{(b)}}{\geq} {}& \frac{\partial f (\sum_{s=1}^{t} A_{sj} x_{sj})}{C-1}\bigg(C^{\frac{\sum_{s=0}^{t} A_{sj} x_{sj}}{Y_j}}-1\bigg),
\end{align}
\end{small}
\noindent where (a) is from the induction hypothesis and (b) follows the inequality $\frac{\log(1+y)}{y} \leq\frac{\log(1+x)}{x}$ when $y \geq x$, and we have chosen $F_{\max} \geq \frac{A_{tj}x_{tj}}{Y_j}, \forall j, \, \forall t$.
\end{proof}

\begin{lemma}{(Properties of DO)}\label{props}
DO produces a primal solution $[x_{tj}], \forall j \in J$, and a dual solution $(\alpha_{tj}, \beta_{tj}),  \forall j \in J$, for all time slots $t$, with the following properties:
\begin{enumerate}
\item The dual solution is feasible for all jobs at all time-slots:
\begin{small}
\begin{eqnarray}
\alpha_{tj} \geq 0, \forall j \in J, \forall t=1, 2, \ldots, T \label{dual1}\\
 \beta_{tj} \geq 0, \forall j \in J, \forall t=1, 2, \ldots, T \label{dual2}
\end{eqnarray}
\end{small}
\item The Primal solution is almost feasible for all jobs at all time slots. The following conditions are satisfied:
\begin{small}
\begin{eqnarray}
\mathbf{x}_t \in \mathbf{R}[t], \forall t=1, 2, \ldots, T \label{primal1} \\
\sum_{t=1}^T x_{tj}  \leq Y_j(1+F_{\max}), \forall j \in J \label{primal2}
\end{eqnarray}
\end{small}
\end{enumerate}
\end{lemma}
We say that the solution is ``almost feasible" since the job size constraint can be slightly violated as seen in \eqref{primal2}. In particular, allocations of a job can exceed the job size by $F_{\max}$, which we assume to be small. We can easily obtain a feasible solution by multiplying all allocations $x_{tj}$ by $(1-F_{\max})$.

\begin{proof}
\eqref{dual1} is straightforward, since by line 4 in the algorithm, $\mathbf{\alpha} \geq 0$. \eqref{dual2} can be shown by noticing that for any job $j$, $\beta_{tj}$ is a non-decreasing geometric series that starts from $0$, thus, $\beta_{tj} \geq 0, \, \forall j \, \forall t$. \eqref{primal1} is also guaranteed by the choice of $\mathbf{x}_t$ by line 4 in the algorithm. \eqref{primal2} is a consequence of Lemma \ref{properties} and Lemma \ref{geom}. Given that a job is completely served, i.e., $\sum_{s=1}^t A_{tj}x_{tj} \geq Y_j$, Lemma \ref{geom} guarantees it's dual variable $\beta_{tj} \geq \partial (f(\sum_{s=1}^t A_{tj}x_{tj} ))$. Lemma \ref{properties} tells us that $\alpha_{tj}= \partial f(\sum_{s=1}^t A_{tj}x_{tj}) \leq \beta_{tj}$. Since DO tries to maximize the inner product $<\alpha -\beta, \sum_{s=1}^{t-1} A_s \mathbf{x}_s+A_t\mathbf{x}>$, having $\alpha_{tj} \leq \beta_{tj}$ implies that $x_{tj}=0$ is optimal. It follows that when a job is completely served, no resources are allocated to that job from thereon. There can only be one iteration where a job can be served over its size, bounding that excess resources by $F_{\max}Y_j$ concludes the Lemma.
\end{proof}
To prove a competitive ratio bound, we will bound the Dual cost in terms of the Primal reward using the next key theorem, and then use the weak duality in Theorem \ref{dualityThm} to obtain our main result.

\begin{theorem} {(Key Theorem)} \label{key}
The dual cost given the Primal-Dual online solution obtained by DO can be bounded as follows:
\begin{small}
\begin{align}
D =&\sum_{t=1}^T \sigma_t(\mathbf{A}_t^T \mathbf{\alpha}_T -\mathbf{\beta}_T) + \mathbf{\beta}^T_{T}Y -\sum_{j=1}^J f_j^*(\alpha_{Tj})\label{terms} \\
& \leq P+P+P\bigg(1+\frac{1}{C-1}\bigg)= P\bigg(3+\frac{1}{C-1}\bigg)
\end{align}
\end{small}
\end{theorem}
To prove the Theorem, we will give three lemmas. Each of those lemmas is to bound one term on the RHS of \eqref{terms}.
\begin{lemma}\label{inner}
For any time slot $t$, DO chooses an allocation that satisfies the following:
\begin{small}
\begin{equation}
<\mathbf{\alpha}_t, \mathbf{A}_t\mathbf{x}_t> \leq \triangle P \label{innerTerm}
\end{equation}
\end{small}
where $\triangle P= \sum_j \triangle P_j=  \sum_{j} f_j(\sum_{s=1}^t A_{tj} x_{tj}) -f_j(\sum_{s=1}^{t-1} A_{tj} x_{tj})$ is the instantaneous utility obtained by DO at time $t$. 
\end{lemma}
\begin{proof}\renewcommand{\qedsymbol}{}
Let $f(\mathbf{y})=\sum_j f_j(y_j)$. By Lemma \ref{properties}, we know that $\mathbf{\alpha}_t \in \nabla f(\sum_{s=1}^t A_s \mathbf{x}_s)$. Substituting in the LHS of \eqref{innerTerm}, and using the concavity of utility function, we get the following
\begin{small}
\begin{equation*}
<\nabla f(\sum_{s=1}^t \mathbf{A}_s \mathbf{x}_s),\mathbf{A}_t x_t> \leq f(\sum_{s=1}^t \mathbf{A}_s \mathbf{x}_s)-f(\sum_{s=1}^{t-1} \mathbf{A}_s \mathbf{x}_s) = \triangle P
\end{equation*}
\end{small}
\end{proof}
\begin{lemma} \label{betas}
The sequence of vectors $[\mathbf{\beta}_1, \mathbf{\beta}_2,\ldots, \mathbf{\beta}_t]$ produced by DO has the following property:
\begin{small}
\begin{equation}
(\beta_{t}-\beta_{t-1})^TY \leq \triangle P \bigg(1+\frac{1}{C-1}\bigg)
\end{equation}
\end{small}
\end{lemma}
\begin{proof}
For any active job $j$, we can bound each element in the LHS inner product as follows:
\begin{small}
\begin{align}
 \begin{split}
 \overset{\mathrm{(a)}}{\leq} &   \beta_{t-1j} Y_j \bigg(\frac{\partial f (\sum_{s=1}^{t} A_{sj} x_{sj})}{\partial f (\sum_{s=1}^{t-1} A_{sj} \mathbf{x}_s)}\bigg(1 + \frac{A_{tj}x_{tj}}{Y_j}\bigg)-1\bigg) \\
 +& \frac{\partial f (\sum_{s=1}^{t} A_{sj}  x_{sj} ) A_{tj} x_{tj}}{(C-1)}
\end{split}&\\
&\overset{\mathrm{(b)}}{\leq}  \partial f (\sum_{s=1}^{t} A_{sj}  x_{sj} )A_{tj}  x_{tj}  \bigg( \frac{\beta_{t-1j}}{\partial f (\sum_{s=1}^{t-1} A_{sj}  x_{sj} )} + \frac{1}{C-1}\bigg)
\\
&\overset{\mathrm{(c)}}{\leq} \triangle P_j (1 +\frac{1}{C-1})
\end{align}
\end{small}
Here (a) is due to the update equation of $\beta$. (b) is obtained by noticing that $\partial f (\sum_{s=1}^{t} A_s x_s) \leq \partial f (\sum_{s=1}^{t-1} A_s x_s)$ by concavity. (c) is because $\beta_{t-1j} \leq \partial f (\sum_{s=1}^{t-1} A_s x_s)$ if $x_{tj} >0$ (since this implies that $\alpha_{t-1j} > \beta_{t-1j}$).
\end{proof}

The next Lemma bounds the last term in \eqref{terms} by bounding the concave conjugate in terms of the original function.
\begin{lemma}\label{conj}
The concave conjugate $f^*(\alpha)$ can be bounded using the term, $\mu_f$ given by 
\begin{small}
\begin{equation}
\mu_f=\sup \{c| f^*(\alpha) \geq c f(u), \alpha \in \partial f(u), u \in K\}
\end{equation}
\end{small} 
for a proper cone $K$, and  $-1 \leq \mu_f \leq 0$.
\end{lemma}
The proof is straightforward from Lemma \ref{properties}. A complete proof of this property is given in Lemma 1 in \cite{eghbali2016designing}.

\begin{proof}{(Theorem \ref{key}):}
The first two terms in \eqref{terms} can be bounded as follows
\begin{small}
\begin{eqnarray}
D' =&\sum_{t=1}^T \sigma_t(\mathbf{A}_t^T \mathbf{\alpha}_T -\mathbf{\beta}_T) + \mathbf{\beta}^T_{T}Y \\
=& \sum_{t=1}^T <\mathbf{\alpha}_T -\mathbf{\beta}_T, \sum_{s=1}^t \mathbf{A}_s \mathbf{x}_s> + \mathbf{\beta}^T_{T}Y\\
 \overset{\mathrm{(a)}}{\leq}& \sum_{t=1}^T < \mathbf{\alpha}_T, \sum_{s=1}^t \mathbf{A}_s \mathbf{x}_s> + \mathbf{\beta}^T_{T}Y \\
  \overset{\mathrm{(b)}}{\leq}& \sum_{t=1}^T < \mathbf{\alpha}_t, \sum_{s=1}^t \mathbf{A}_s \mathbf{x}_s> + \mathbf{\beta}^T_{T}Y\\
 \overset{\mathrm{(c)}}{=}& \sum_{t=1}^T < \mathbf{\alpha}_t, \sum_{s=1}^t \mathbf{A}_s \mathbf{x}_s> + \sum_{t=1}^T(\mathbf{\beta}_{t}-\mathbf{\beta}_{t-1})^TY \\
  \overset{\mathrm{(d)}}{\leq} &  \sum_{t=1}^T \triangle P (2+ \frac{1}{C-1}) =P(2+\frac{1}{C-1})
\end{eqnarray}
\end{small}

\noindent where (a) is because $\beta_T \geq 0$, so dropping the term $<-\mathbf{\beta_T}, \sum_{s=1}^t A_s x_s>$ can only increase the objective. (b) is because $\alpha_t \geq \alpha_T$, by Lemma \ref{properties} and the concavity of the function, thus decreasing gradients. (c) is true due to telescoping and the fact that $\beta_0=0$. (d) holds by substituting the bounds from Lemmas \ref{inner} and \ref{betas}.\\
Finally by \eqref{terms}, we have $D=D'-\sum_{j=1}^J f_j^*(\alpha_{Tj})$. We can bound that extra $-\sum_{j=1}^J f_j^*(\alpha_{Tj})$ term on the RHS by $P$ utilizing Lemma \ref{conj}. Adding that bound to the bound on $D'$ concludes the proof.
\end{proof}

\begin{corollary}\label{DOcomp}
The online solution found by DO is $(3+\frac{1}{C-1})$-competitive.
\end{corollary}
We note two things about our results
\begin{enumerate}
\item To guarantee primal feasibility, the BS can multiply the resource allocation solution by $(1-F_{\max})$ at each time slot. This adds an extra factor to the Competitive Ratio making the algorithm $(3+\frac{1}{C-1})(1-F_{\max})$-competitive.
\item Practically, we expect $F_{\max}$ to be small as the job service times have a slower time scale than the scheduling job completion time scale. Thus we expect $F_{\max} \to 0$ making the algorithm approximately $3+\frac{1}{e-1}$-competitive.
\end{enumerate}
\subsection{Lightweight Algorithm}
The complexity of the DO Algorithm can be further reduced by splitting the saddle point problem in line 4 into two separate steps as follows:
\begin{small}
\begin{equation*}
\underset{\mathbf{x} \in \mathbf{R}[t]}{\max}  \, <\alpha_{t-1} -\beta_{t-1},A_t \mathbf{x}>, \,\,\,\,\,
\alpha_{tj} \in  \partial(f_j(\sum_{s=1}^t A_{sj} x_{sj}))
\end{equation*}
\end{small}
This approximation was proposed in \cite{eghbali2016designing} in the context of online bipartite matching. This formulation approximates the saddle point problem with a Linear Programming problem, reducing complexity. However, the price of this reduction in complexity is an increase in the constant-competitive ratio bound that depends on the specific utility function gradients (\cite{eghbali2016designing} analyzes this penalty in the bipartite matching problem). We will show using numerical simulations that this approximation retains the good performance of the DO algorithm.

\section{Stochastic Setting with timely throughput constraints}       
Although the job/reward formulation in \eqref{formulation} has been used extensively in modeling scheduling with hard deadlines, for example \cite{zheng2016online}\cite{lucier2013efficient}\cite{devanur2018primal} , a formulation that aims to maximize total rewards of jobs is susceptible to unfairness. For example, the BS can maximize the sum of rewards by consistently allocating resources to a nearby user experiencing better channels all the time. This phenomenon was reported in previous works \cite{liu2001opportunistic} and is further validated by simulations. Furthermore, the results in the previous section hold for adversarial models, designed for ``worst case" inputs. In practice however, both the job arrivals processes and the rate regions are stochastic. We propose a new model to deal with those two issues that have the following extra assumptions:

\begin{assumption}\label{assumption1} 
\begin{enumerate}[wide, labelwidth=!, labelindent=0pt]
\, \, \, \item We assume a frame structure: At the beginning of a frame of size $D$, some jobs arrive to the BS to be transmitted to users. By the end of the frame after $D$ slots, all jobs expire, and the system is empty.  Note that jobs can still have different deadlines as long as they are all upper bounded by $D$. The frame structure has been extensively used in modeling deadline-constrained traffic \cite{hou2014scheduling}  \cite{jaramillo2011optimal}\cite{deng2017timely}. This assumption has been shown to adequately approximate practical scenarios, while enabling the design of efficient scheduling algorithms with deterministic bounds on delay.
\item We assume that there are $l$-job classes with specified deadlines, reward functions, and sizes. Each of these $l$-classes arrive at the beginning of the frame according to an i.i.d arrival process $\mathcal{A}_k$. We assume that the number of the new jobs arriving at the beginning of a frame can be deterministically bounded, i.e., $(m(t)) \leq M$, where $m(t)$ is a random variable representing the number of active jobs at time $t$. 
\item We assume that the instantaneous rate region $\mathbf{R}[t]$ is sampled every time slot from a set of finite convex regions in an i.i.d manner unknown to the BS. The realization of rate regions over a frame is denoted as $\mathcal{R}_k$.
\end{enumerate}
\end{assumption}

The new formulation is presented in \eqref{formulation2}. Our goal now is to maximize the long-term average expected rewards over frames $k=1, \ldots, K$. We denote the jobs that arrive at frame $k$ as $J^k$. In \eqref{timely2}, we introduce a new constraint to guarantee fairness by ensuring that every user gets an expected \textbf{timely-throughput} higher than $\delta_n$. Timely-throughput is the amount of traffic delivered within the deadline over a period of time. It has been used extensively to analyze networks with real-time traffic \cite{hou2014scheduling}\cite{lashgari2013timely}. The function $U(\,)$ simply maps the job $j$ to its intended user $n$.

\vspace{-0.3cm}
\begin{small}
\begin{subequations}\label{formulation2}
        \begin{align}
    & \underset{\mathbf{x}_1, \ldots, \mathbf{x}_t}{\max}
    & & \underset{K \to \infty}{\liminf} \frac{1}{K} \sum_{k=1}^K \mathbb{E}\bigg\{ \sum_{j \in  J^k} f_j(\sum_{t=kD}^{(k+1)D-1} A_{tj} x_{tj})\bigg\} \label{objective2}&&& \\
    & \text{subject to}
    & &  \underset{K \to \infty}{\liminf} \frac{1}{K} \sum_{k=1}^K \mathbb{E}\bigg\{ \sum_{j \in  J^k \cap U(j)=n} \sum_{t=kD}^{(k+1)D-1} A_{tj} x_{tj}\bigg\} \geq \delta_n \label{timely2}&&&\\
    &&& \sum_{t=1}^T x_{tj} \leq Y_j, \;  \forall j \label{budget2}&&& \\ 
    &&& \mathbf{x_t} \in \mathbf{R}[t], \;  \;  \forall t= 1,2, \ldots, T.\label{capacity2}&&&
\end{align}
    \end{subequations}
    \end{small}   
We refer to a random realization of job arrivals and rate regions over a frame as $q$. The optimization problem \eqref{formulation2} can be solved by a stationary scheduler that maps $q=\{\mathcal{A}_k, \mathcal{R}_k\}$ into the set of feasible actions over the frame:\\ $\chi=\{ x | \sum_{t=kD}^{(k+1)D-1} x_{tj} \leq Y_j, \;  \forall j \in J^k, \mathbf{x_t} \in \mathbf{R}[t],   \forall t= kD, \ldots, (k+1)D-1\}$ with probabilities  $p_{q \chi}$. Thus, the optimal solution can be derived by finding the probabilities $p_{q \chi}$ that solve \eqref{stat}. This is practically infeasible as the probabilities $q$ are typically unknown to the BS. Even if the probabilities were known, the BS needs to non-causally know the rate regions for the entire frame. This motivates us to extend our DO algorithm for the stochastic setting to solve \eqref{stat} and derive performance guarantees. 

\vspace{-0.3cm}
\begin{small}
 \begin{subequations}\label{stat}
   \begin{align}
    & \underset{p_{q\chi}}{\max}
    & & \sum_{q} \nu_q  \int_{\chi \in X_q} p_{q \chi} \sum_{j} f_j(\sum_{t=KD}^{(K+1)D-1} A_{tj} x_{tj}) d\chi 
 \\
    & \text{subject to}
    & & \sum_{q} \nu_q  \int_{\chi \in X_q} p_{q \chi} \sum_{j|U(j)=n} \sum_{t=KD}^{(K+1)D-1} A_{tj} x_{tj} d\chi \geq \delta_n \label{fairness}\\  
    &&& \int_{\chi \in X_q} p_{q \chi} d\chi=1 \, \, \forall q \\
    &&& \int_{\chi \in X_q} p_{q \chi} \geq 0 \, \, \forall q 
\end{align}
    \end{subequations}
    \end{small}   
    \vspace{-0.5cm}
\subsection{Virtual Queue Structure}
To deal with the new timely throughput constraints \eqref{timely2} for each user $n$, we define a virtual queue that records constraint violations. For every frame, the amount of unserved work under the $\delta_n$ requirement, $\delta_n - \sum_{t=1}^T A_{tj} x_{tj}$ is added to the queue, i.e., the queue is updated as follows:

\begin{small}
\begin{equation}
Q_{n}[k+1]=(Q_n[k]+\delta_n - \sum_{j \in J^{k} \cap U(j)=n} \sum_{t=kD}^{(k+1)D-1} A_{tj} x_{tj})^+ ,\label{evolution}
\end{equation}
\end{small}
where $(x)^+=\max(0,x)$. There are two time-scales at play here. First, the slower frame-level time scale. At the beginning of a frame, jobs arrive and by the end of the frame, those jobs expire. Second, the faster slot level time-scale, where the channels change and the BS allocates rates $\mathbf{x}$. Each frame consists of $D$ time slots where all jobs are guaranteed to expire by the end of the frame by Assumption \ref{assumption1}.\\
Virtual queues are used to analyze the time-average constraint violation for a given scheduling policy. It can be shown that stability of the virtual queue ensures that the constraint is satisfied in the long term. We state that well-known result as a Lemma without proof (The proof is simple and can be found in \cite{neely2010stochastic} \cite{tan2012online})

\begin{lemma}
For any user $n$, the virtual queue length upper bounds the constraint violation at all times as follows:
\begin{small}
\begin{equation}
\frac{Q_n[K]}{K}-\frac{Q_n[0]}{K} \geq \delta_n - \frac{1}{K}\sum_{k=1}^K \sum_{j \in J^{k} \cap U(j)=n}\sum_{t=kD}^{(k+1)D-1} A_{tj} x_{tj}
\end{equation}
\end{small}

Furthermore the mean rate stability defined as: 
\begin{small}
\begin{equation}
\lim_{K \to \infty}   \frac{\mathbb{E}(Q_n[K])}{K}=0
\end{equation}
\end{small}
implies that the constraint \eqref{timely2} is satisfied in the long-term.
\end{lemma}
\subsection{D Look-ahead Algorithm}
Before explaining our algorithm, we present and analyze a non-causal frame-based algorithm that we refer to as the D look-ahead algorithm. The benefits of this hypothetical algorithm are two-fold: First, it guides our design of the practical LFDO algorithm in the next section, and second, it will be crucial in analyzing the performance of LFDO.\\
The D look-ahead algorithm observes the jobs $J^k$ at the beginning of the frame and non-causally observes all rate regions over the frame $\mathbf{R}[k], \mathbf{R}[k+1], \ldots, \mathbf{R}[k+D-1]$, and allocates rates $\mathbf{x'}$ of jobs over the frame $k$ by solving the following optimization problem:

 \vspace{-0.2cm}
\begin{small} 
 \begin{subequations}\label{Dlookahead}
        \begin{align}
    & \underset{\mathbf{x}_{kD}, .., \mathbf{x}_{(k+1)D-1}}{\max}
     & \begin{split} &V\sum_{j \in  J_k} f_j(\sum_{t=kD}^{(k+1)D-1}A_{tj} x_{tj}) \\
    &+ \sum_{n=1}^N Q_n[k]\bigg(\sum_{j|U(j)=n} \sum_{t=kD}^{(k+1)D-1} A_{tj} x_{tj}\bigg) \label{DLAreward} \end{split} \\
    & \text{subject to}
    & &\sum_{t=1}^T x_{tj} \leq Y_j, \;  \forall j  \in J_k\\ 
    &&& \mathbf{x_t} \in \mathbf{R}[t], \;  \;  \forall t \in [kD,(k+1)D-1] 
\end{align}
 \end{subequations}
 \end{small}
\noindent where $V$ is a free parameter that will be used to manage the trade-off between the timely-throughput short-term constraint violation and total reward achieved by the algorithm. The D look-ahead algorithm is essentially a version of the well-known drift-plus-penalty algorithm introduced in \cite{neely2010stochastic}  that has been used extensively in stochastic constrained optimization problems, where a queue structure can be used to deal with long-term constraints.\\
 To simplify the notation, we will refer to the frame $k$ D look-ahead reward and timely throughput, respectively as follows:
 
 \vspace{-0.2cm}
 \begin{small}
 \begin{align} 
P'[k] &=  \sum_{j \in  J_k} f_j(\sum_{t=kD}^{(k+1)D-1}A_{tj} x'_{tj}) \label{not1}\\ 
b_n'[k] &=  \bigg(\sum_{j \in J_k|U(j)=n} \sum_{t=kD}^{(k+1)D-1} A_{tj} x'_{tj}\bigg) \label{not2}
\end{align}
\end{small}
 We define the quadratic Lyapunov function $L(\mathbf{Q}[t])= \frac{1}{2} \sum_{n=1}^N Q_n^2[t]$. We also define the one step Lyapunov drift and bound it as follows: 
 \begin{small}
 \begin{multline} 
\triangle \mathbf{\Theta}(\mathbf{Q}) =\mathbb{E}(L(\mathbf{Q}[k+1])-L(\mathbf{Q}[k])|\mathbf{Q}[k]= \mathbf{Q}) \\ \leq B + \sum_{n=1}^N Q_n[k](\delta_n-b_n[k]) \label{drift}
\end{multline}
\end{small}
\noindent where $B$ is a bound on the term $\mathbb{E}\big((\delta_n-b_n[k])^2\big)$, which is guaranteed to exist due to the boundedness of the number of jobs and the job sizes. It can be seen that the D Look-ahead algorithm in \eqref{Dlookahead} attempts to maximize the reward while minimizing the drift (and subsequently the queue lengths), using the parameter V to manage the trade-off. We are now ready to state the theorem that bounds the performance of the D look-ahead theorem.

\begin{theorem} \label{Dlk}
Suppose there exists a solution that can achieve a timely throughput strictly greater than $\delta_n+\epsilon$, for some $\epsilon>0$, for all users. Under the D look-ahead solution, the queues $Q_n, \forall n$  are mean-rate stable, and the following holds:

\begin{small}
\begin{align}
\liminf_{K \to \infty} \frac{1}{K}\sum_{k=1}^K \mathbb{E}(P'[k]) &\geq P^* -\frac{B}{V} \label{Sreward}\\
\limsup_{K \to \infty}\frac{1}{K}\sum_{k=1}^K  \sum_{n=1}^N \mathbb{E}(Q_n[k]) &\leq \frac{B+VM f_{\max}(Y_{\max})}{\epsilon} \label{Squeue}
\end{align}
\end{small}
\end{theorem} 

Before giving the proof, we point out that Theorem \ref{Dlk} shows that the D look-ahead algorithm can be made arbitrarily close to OPT by increasing $V$, at the cost of increasing the queue lengths, which implies higher short term violation of the timely throughput constraint. The main assumption of the theorem is a mild assumption that a strictly feasible solution exists, i.e., timely throughput constraints cannot be set arbitrarily and must be strictly feasible under some solution. This corresponds to the ``Slater conditions" that are essential to applying the Lyapunov arguments \cite{neely2010stochastic}.
\begin{proof}
Let the reward achieved and the amount of traffic served by the D look-ahead be denoted by $P'[k]$ and $b'[k]$ as in \eqref{not1} and \eqref{not2}, respectively. Similarly OPT achieves $P^*[k]$ and $b^*[k]$. By the maximization in \eqref{DLAreward}, we have:
\begin{small}
\begin{equation}
\sum_{n=1}^N Q_n[k](\delta_n-b_n'[k])-VP'[k] \leq \sum_{n=1}^N Q_n[k](\delta_n-b_n^*[k])-VP^*[k] \label{DplusP}
\end{equation} 
\end{small}
\noindent for any frame instance, thus, the inequality holds for the conditional expectation given the queue lengths equal to $\mathbf{Q}$. Noting the definition of the drift in \eqref{drift}, we can bound $ E(\triangle \mathbf{\Theta}'(\mathbf{Q}))-V\mathbb{E}(P'[k]|\mathbf{Q})$  by the following:
 
\begin{small}
\begin{IEEEeqnarray}{lCr}
\leq & E(\triangle \mathbf{\Theta}^*(\mathbf{Q}))-V\mathbb{E}(P^*[k]|\mathbf{Q}) \\
\overset{\mathrm{(a)}}{\leq} & B+\sum_{n=1}^N Q_n[k]\mathbb{E}(\delta_n-b^*_n[k]) -V\mathbb{E}(P^*[k]|\mathbf{Q}) \label{Dineq}\\
\overset{\mathrm{(b)}}\leq & B-V\mathbb{E}(P^*[k]|\mathbf{Q})
\end{IEEEeqnarray}
\end{small}
\noindent where (a) is by the bound in \eqref{drift}, and (b) is because the optimal stationary solution satisfies the constraint in expectation independent of $\mathbf{Q}$.\\
Taking the expectation over $\mathbf{Q}$ and taking the time average over all the frames, we can use telescoping sums to arrive at the key equation

\vspace{-0.3cm}
\begin{small}
\begin{equation}
L(Q[k])-L(Q[0]) -\frac{V}{K} \sum_{k=1}^K \mathbb{E}(P'[k]) \leq B-VP^*
\end{equation} 

\end{small}
Noting that $L(Q[k])$ is non-negative, initializing $L(Q[0])$ to $0$, and rearranging the sum gives \eqref{Sreward}.\\
To prove \eqref{Squeue}, we can follow the same steps by comparing the solution produced by the D Look-ahead algorithm to another solution that can strictly satisfy the constraint \eqref{timely2}, i.e., $\mathbb{E}(\delta_n-b_n[k]) < -\epsilon, \forall n$, for some $\epsilon >0$. This solution is guaranteed to exist by the assumption in the Theorem statement. We denote the reward of that solution as $P(\epsilon)$. Repeating the same steps up to \eqref{Dineq} we get the following inequality: 

\begin{small}
\begin{multline}
\mathbb{E}(\triangle \mathbf{\Theta'}(\mathbf{Q}[k]))-V\mathbb{E}(P'[k]) \\ \leq  B- \epsilon \sum_{n=1}^N \mathbb{E}(Q_n[k])-V\mathbb{E}(P(\epsilon)[k]|\mathbf{Q})
\end{multline}
\end{small}
Similar to last part, we can take the time average over frames and telescope to get
\begin{small}
\begin{equation*}
\frac{1}{K}\sum_{k=1}^K \sum_{n=1}^N \mathbb{E}(Q_n[k]) \leq  \frac{B+V\mathbb{E}(P(\epsilon))-\mathbb{E}(P'))+ \mathbb{E}(L(Q[0])}{\epsilon}
\end{equation*}
\end{small}
Bounding $P(\epsilon)$ by the $M f_{\max}(Y_{\max})$, the maximum achievable reward over the frame, and taking the limit gives \eqref{Squeue}.
\end{proof}
\subsection{Long-term Fair Deadline Oblivious (LFDO) Algorithm}
We are now ready to present our modified deadline oblivious algorithm that can satisfy long-term timely throughput constraints.
\begin{small}
\begin{algorithm}
\caption{Long-term Fair Deadline Oblivious (LFDO) Algorithm}
\SetKwInput{kwInit}{Initialize}
\SetKwInput{kwInitF}{Initialize Frame}
\kwInit{At $k=0$, set $Q_n[k]=0, \; \forall n$}\
\For{$k=1$ \KwTo $K$}
{
\kwInitF{Receive jobs at the beginning of the frame}
\For{$t=kD$ \KwTo $(k+1)D-1$}
{ 
Perform the DO algorithm with the modified job reward function, $g_j$:
\begin{small}
\begin{equation}
g_j(x)=Vf_j(x)+\sum_{n=1}^N \mathbbm{1}(U(j)=n)Q_n[k]A_{tj}x_{tj} \label{modreward}
\end{equation}
\end{small}
}
Update the queues according to \eqref{evolution}
}
\end{algorithm}
\end{small}
\vspace{-0.3cm}

As can be seen in Algorithm 2, LFDO is a modified version of the DO algorithm incorporating long term timely throughput guarantees. This is done by building on the virtual queue idea shown in the D look-ahead solution. There are two time scales at play here:
\begin{itemize}[leftmargin=2mm]
\item Frame time scale: The slower time scale where virtual queues are updated according to the LFDO solution over the frame duration.
\item Slot time scale: The faster time scale where the DO algorithm operates. Every frame length acts as the ``horizon" for the DO algorithm. At the beginning of the frame, DO re-initializes to serve the jobs that belong to that frame.
\end{itemize} 
The reward function in line 3 has been modified to add the user queue length information to the job reward function. This follows the drift-plus-reward maximization used to obtain the D look-ahead solution in \eqref{Dlookahead}. The difference is, unlike the D look-ahead solution, LFDO does not know the future rate regions. Thus, on time-slot scale, LFDO uses the primal-dual optimization used for DO with the modified reward.\\
We are now ready to combine our results of the DO algorithm performance and the D look-ahead solution performance to obtain a powerful performance result for the LFDO algorithm in the next theorem
\begin{theorem} \label{LFDOperf}
Under the LFDO Algorithm in the stochastic setting, all queues are mean rate-stable. Furthermore, the expected reward and the expected queue length can be bounded as follows:
\begin{small}
\begin{align}
\liminf_{K \to \infty} \frac{1}{K}\sum_{k=1}^K \mathbb{E}(P[k]) &\geq \frac{1}{\gamma}(P^*-\frac{B}{ V}) \label{LFDOreward}\\
\limsup_{K \to \infty}\frac{1}{K}\sum_{k=1}^K  \sum_{n=1}^N \mathbb{E}(Q_n[k]) &\leq \frac{\gamma( B+ VMf_{\max}(Y_{\max}))}{\epsilon} \label{LFDOqueue}
\end{align}
\end{small}
\noindent where $\gamma$ is the Competitive Ratio achieved by the DO algorithm.
\end{theorem} 

  This result asserts that the LFDO algorithm maintains its power when moving from the adversarial to the stochastic setting. In particular, LFDO satisfies the timely throughput constraint by \eqref{LFDOqueue}. This comes at the cost of a larger queue length compared to the D Look-ahead algorithm. Similarly, the LFDO algorithm can be made arbitrarily close to achieve a $\frac{1}{\gamma}$-fraction of the stationary optimal reward, where $\gamma$ is the constant competitive ratio achieved by the DO algorithm in Corollary \ref{DOcomp}. This reduction of reward compared to the D Look-ahead algorithm is due to the non-causality advantage that the D look-ahead algorithm has over the LDFO algorithm. However, Theorem \ref{LFDOperf} shows that LFDO guarantees each user a long-term stochastic timely throughput while achieving a \textit{constant fraction} of the long-term optimal reward independent of the problem size. Proving \ref{LFDOperf} is straight-forward given the machinery we have already built. 
  
\begin{proof}
We prove the theorem by applying the key Theorem \ref{key} with reward-plus-drift function, over the frame length. Since LFDO maximizes the sum of $g_j(\,)$ functions over every frame, Theorem \ref{key} guarantees that LFDO  achieves a modified reward that is at least a $\frac{1}{\gamma}$-fraction of the reward achieved by any offline solution. Thus, we can relate the reward-plus-drift achieved by the LFDO (in the LHS) to the one achieved by the D look-ahead (in the RHS) as follows:

\vspace{-0.2cm}
\begin{small}
\begin{multline}
\gamma \bigg(\sum_{n=1}^N Q_n[k](\delta_n-b_n[k])-VP[k]\bigg)\\ 
 \leq  \sum_{n=1}^N Q_n[k](\delta_n-b_n'[k])-VP'[k] \label{LongDplusP}
\end{multline} 
\end{small}
The rest of the proof is straight-forward and follows exactly the steps of the proof of Theorem \ref{Dlk}
\end{proof} 
\section{Numerical Results}
We assess the performance of our proposed algorithms with numerical simulations. We first show that Lightweight DO tracks the offline solution very closely, outperforming several existing algorithms in the literature. We compare DO to a state of the art algorithm that was proposed in \cite{zheng2016online} in the datacenter context. We call that algorithm ``Primal" since it is also a deadline oblivious algorithm that only relies on the primal but not the dual updates to determine the allocation. Despite being a datacenter algorithm, Primal also attempts to maximize total partial job rewards, and is therefore comparable to DO (although the competitive ratio results were derived for a wired setting only). We also compare the performance against the Earliest-Due Date (EDD) that was analyzed in \cite{dua2007downlink} for packets as a benchmark, and a greedy algorithm that was proposed for linear reward functions in \cite{agarwal2002base}.\\ 
\textbf{Setup:} We simulate a downlink cell with three users (we chose a small number of users to enable the offline solver to run with reasonable memory requirements). Each time slot, for each user, a new jobs arrive to the BS intended to that user with probability $p$. Thus, $p$ represents the \textit{traffic intensity} of the system. The job sizes are uniformly distributed between 5 and 25 units. Each job has a random deadline uniformly distributed between 2 and $D_{\max}$ time slots. $D_{\max}$ represents the \textit{laxity} of the system. Smaller $D_{\max}$ means tighter deadlines. Large $D_{\max}$ implies more variety in traffic. The instantaneous rate region is generated by sampling a uniformly random distribution for each user, then taking the convex hull of those user samples. The resultant rate region is non-orthogonal. Finally, each job has a random reward function of $\frac{v(0.1+x)^{(1-\psi)}}{1-\psi}$, where $v$ and $\psi$ are uniformly distributed between 0 and 1.\\
\textbf{Performance of DO: } In Fig. \ref{traffic}, we plot the performance of different algorithms and OPT while varying traffic intensity, $p$. It is clear that DO tracks the OPT very closely, confirming our premise that Deadline Oblivious scheduling is efficient for real-time traffic. DO consistently performs $8-15\%$ better than Primal at a \textit{lower complexity}, since lightweight DO has the complexity of a linear program while primal solves a general convex program. Greedy and EDD perform significantly worse. In Fig. \ref{laxity}, we vary $D_{\max}$ between 2 and 40 time slots to simulate different workloads. The results are similar to the previous figure with DO closely tracking OPT. Interestingly, there is a slight performance degradation for very small values of $D_{\max}$ when deadlines are very tight. This is consistent with our findings regarding the dependence of competitive ratio bound on $F_{\max}$, the job-size-to-capacity ratio.\\
\begin{figure} \label{Sim3}
 \setlength{\belowcaptionskip}{-10pt}
    \centering
  \subfloat[Varying $p$ \label{traffic}]{%
       \includegraphics[width=0.5\linewidth,height=3cm]{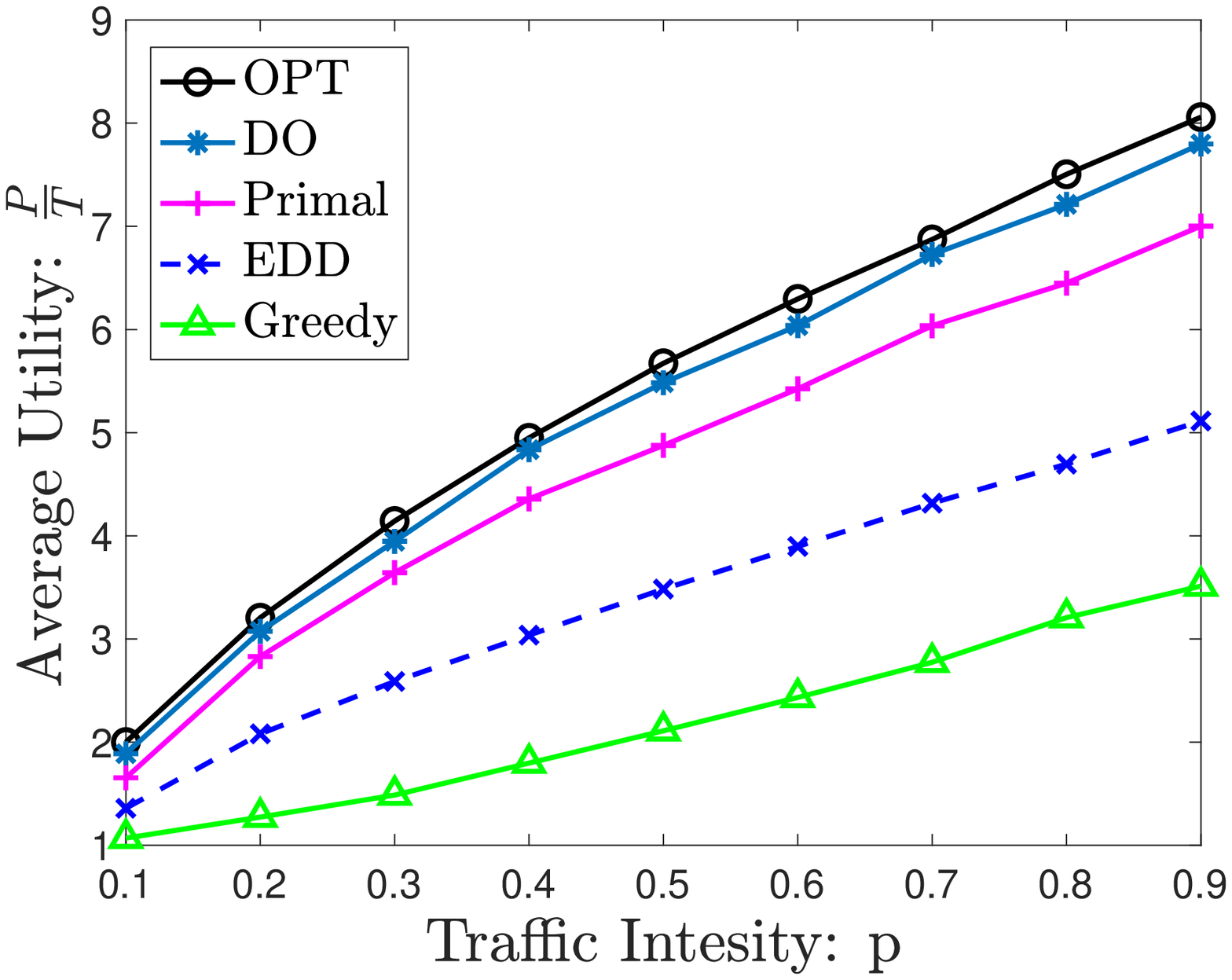}}
    \hfill
      \subfloat[Varying $D_{\max}$ \label{laxity}]{%
        \includegraphics[width=0.5\linewidth,height=3cm]{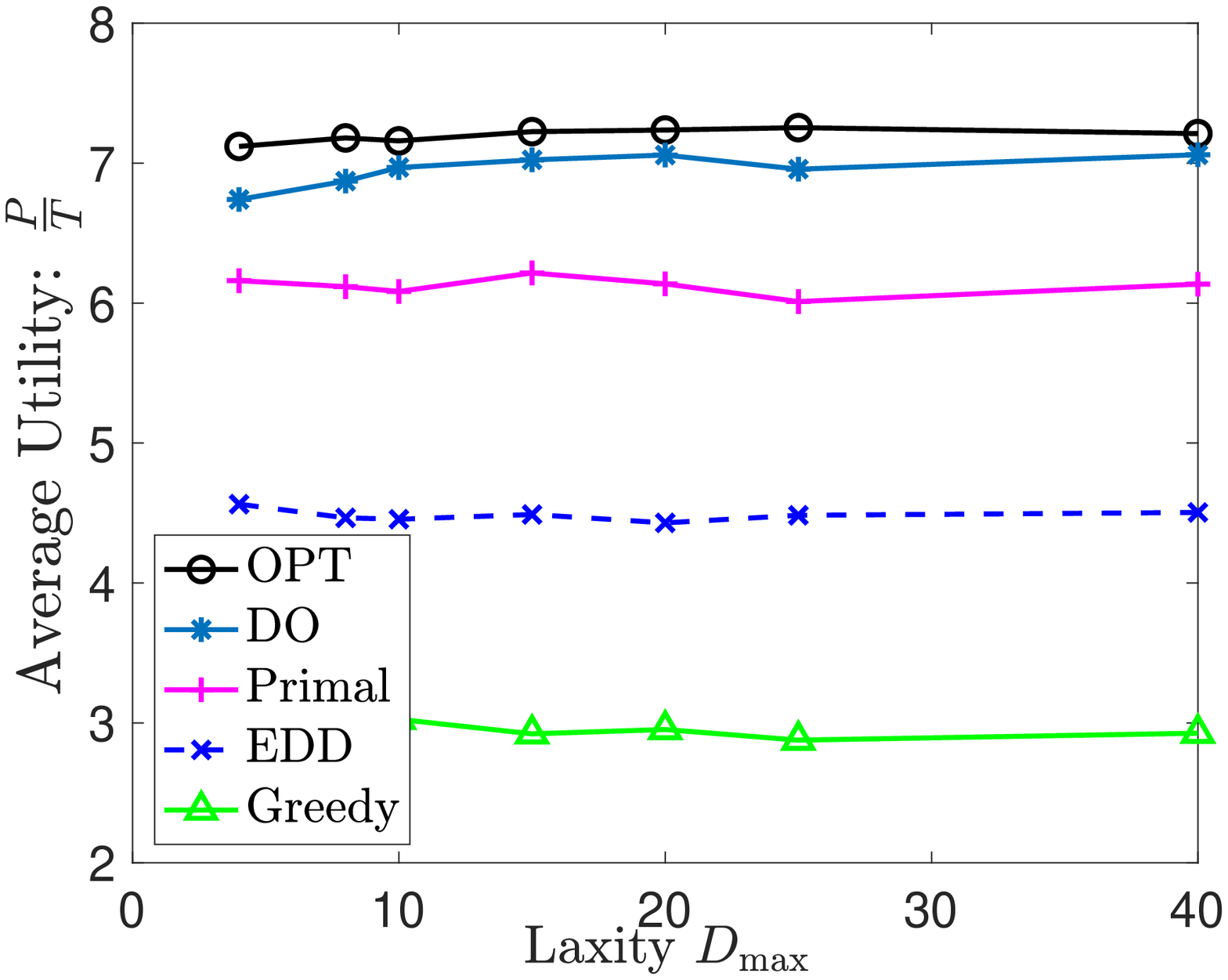}}
  \caption{Comparison of performance of different algorithms}
  \label{fig1} 
\end{figure}
\noindent \textbf{Performance of LFDO:} In Fig. \ref{Sim3}, we simulate the system for five users. We set up the simulation, such that User 1 consistently gets low feasible rates compared to other users. In particular, we sample the random rates such that User 1 can get a maximum timely throughput of 0.05, and other users can get up to 0.5 timely throughput. The instantaneous rate region is the convex hull of random rates. We set a minimum timely throughput constraint of 0.045, thus pushing the system to the boundary of the ``capacity region" by forcing User 1 to operate very close to its upper limit. In Fig. \ref{1a}, we show the timely throughput of all users under DO. Since DO tries to maximize reward with no regard to timely throughput constraints, we see that User 1 converges to a timely throughput well below the requirement. In Fig. \ref{1b}, we run LFDO for the same system with $V=1$. Despite the improvement over DO, the convergence to the required timely-throughput level is slow since virtual queues are allowed to backlog before being cleared. In Fig. \ref{1c}, we set $V=0.1$ emphasizing the importance of timely throughput constraints. The result is that User 1 can now satisfy the constraint with fairly quick convergence at the expense of slightly decreased reward (within $95\%$ of DO reward). This outlines the previously stated trade-off between the reward and the timely throughput guarantees.
\begin{figure} 
  \subfloat[DO\label{1a}]{%
       \includegraphics[width=0.33\linewidth]{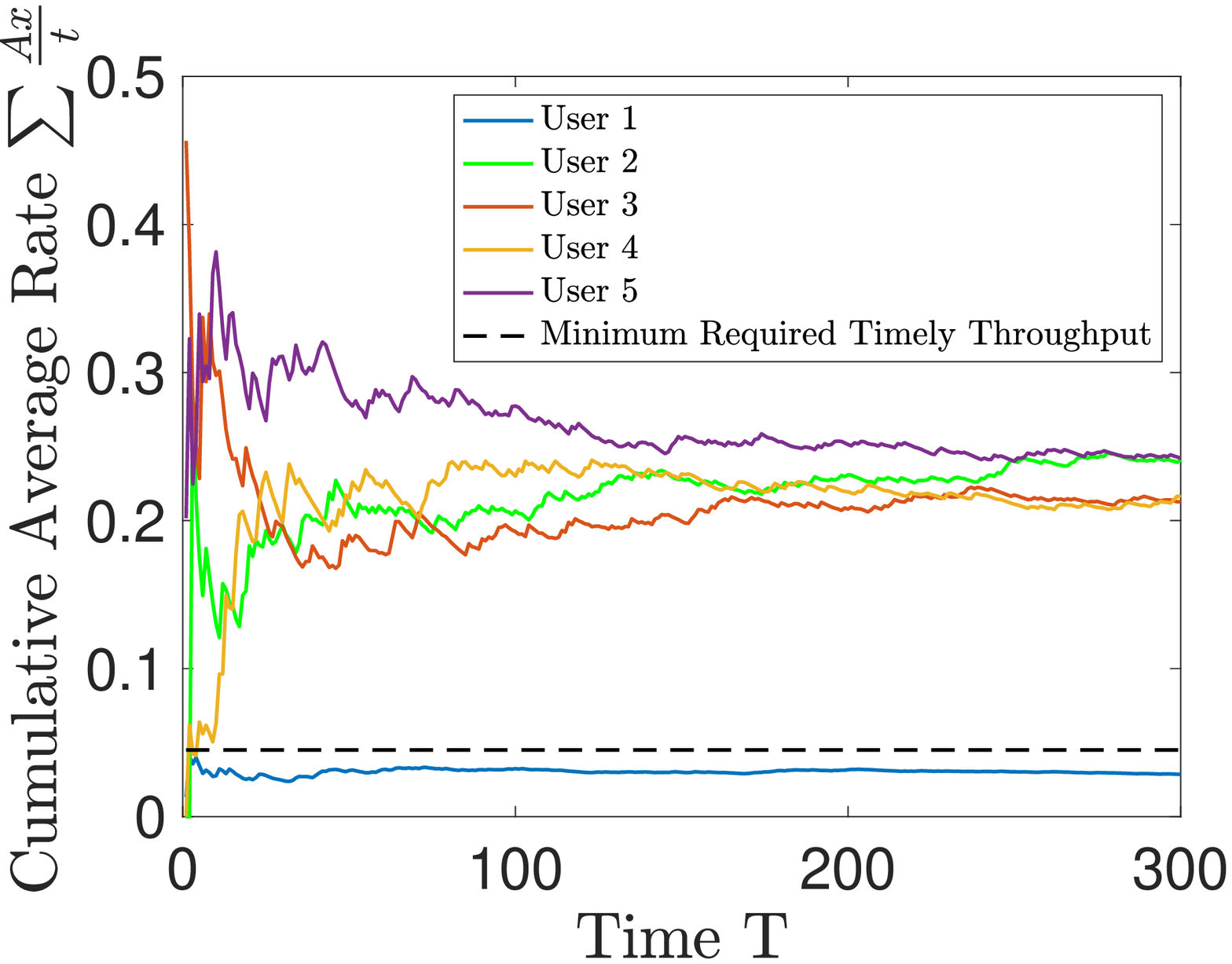}}
\subfloat[LFDO, $V=1$\label{1b}]{%
        \includegraphics[width=0.33\linewidth]{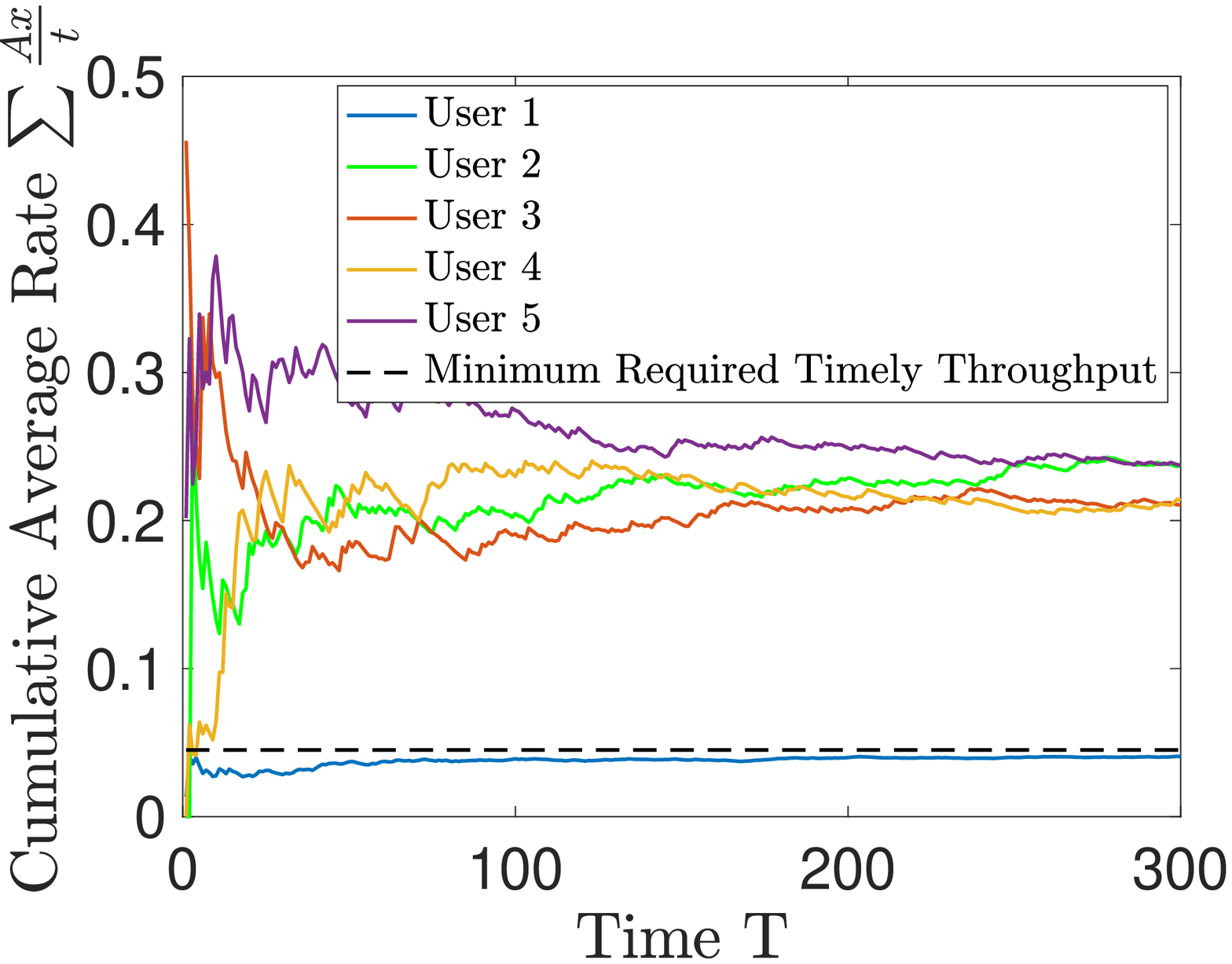}}
  \subfloat[LFDO, $V=0.1$\label{1c}]{%
        \includegraphics[width=0.33\linewidth]{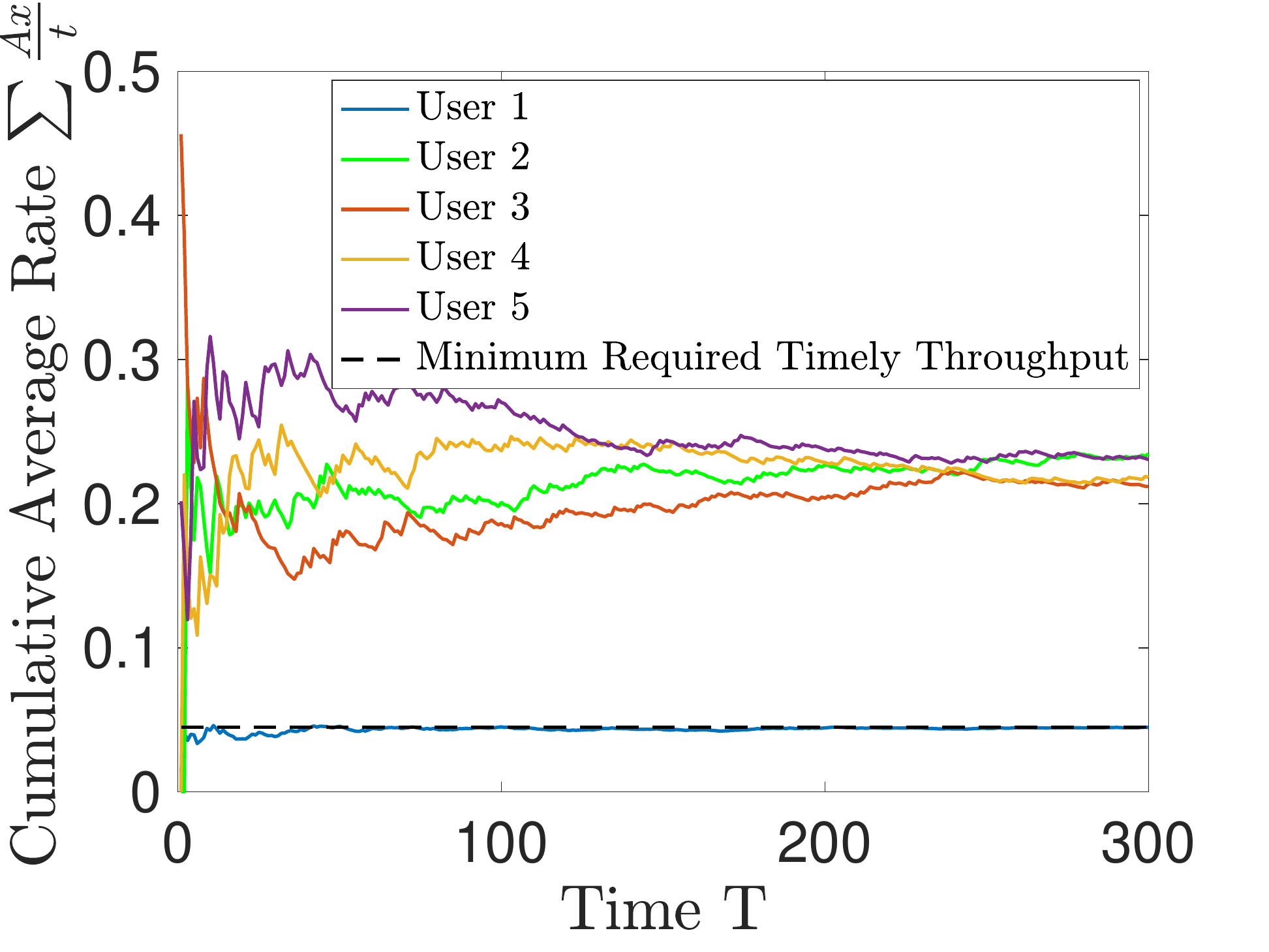}}
  \caption{Resource allocation per user under DO and LFDO}
  \label{Sim3} 
\end{figure}
\section{Conclusion and Future Work}
We have studied the problem of resource-allocation of low-latency bandwidth-intensive traffic. We have formulated the problem as an online convex optimization problem, developed a low-complexity Primal-Dual DO algorithm and derived its competitive ratio. We have demonstrated that our algorithm is efficient \textit{and does not rely on deadline information}. We have also proposed the LFDO algorithm, that modifies DO to satisfy long-term stochastic timely-throughput constraints. We have shown via simulations that our proposed algorithms tracks the offline optimal solution very closely and performs better than existing solutions. In the future work, we aim to understand the properties of DO algorithm better, for example, we aim to analyze how many jobs are served to their completion. This will enable us to expand the algorithm to serve traffic that must be served to completion as well as traffic that has the partial utility property. We aim to develop our work to take the unreliability of wireless channels and inaccurate channel estimations into account. We also plan to test our algorithm with a real-time setup through the variety of traffic seen in 5G networks.  

\section{Acknowledgment}
This work has been supported in part by National Science Foundation awards CNS-1618566, CNS-1719371, CNS-1409336, CNS-1731698,
CNS-1814923, and from the Office of Naval Research award N00014-17-1-2417.

\bibliographystyle{IEEEtran}  
\bibliography{main}  

\end{document}